\documentclass[a4,12pt]{article}%

\usepackage{graphicx}
\usepackage{graphics}

\usepackage{amsmath}
\usepackage{amsfonts}
\usepackage{amssymb}
\usepackage{enumerate}
\usepackage{xcolor}
\usepackage{tikz}
\usepackage{textcomp} 
\definecolor{Green}{rgb}{0.0, 0.5, 0.0}

\definecolor{dgreen}{rgb}{0.0, 0.5, 0.0}
\newtheorem{theorem}{Theorem}[section]

\newtheorem{conjecture}[theorem]{Conjecture}

\newtheorem{definition}[theorem]{Definition}

\newtheorem{lemma}[theorem]{Lemma}

\newtheorem{notation}[theorem]{Notation}

\newtheorem{remark}[theorem]{Remark}

\newenvironment{proof}[1][Proof]{\textbf{#1.} }{\ \rule{0.5em}{0.5em}}

\newcommand{\numbercellong}[2]

\usepackage{dirtytalk}

\usepackage[english]{babel}

\usepackage[letterpaper,top=2cm,bottom=2cm,left=3cm,right=3cm,marginparwidth=1.75cm]{geometry}

\usepackage{amsmath}
\usepackage{graphicx}
\usepackage[colorlinks=true, allcolors=blue]{hyperref}

\title{A Weak Structural Form of Commutative Equivalence in
Finite Codes}
\author{Dean Kraizberg\thanks{School of Mathematical Sciences, Tel Aviv University, deank@mail.tau.ac.il}}

\begin{document}

\maketitle
\begin{abstract}
   We investigate the structural relationship between prefix-free codes over the binary alphabet and a class of unlabeled rooted trees, which we call \emph{symmetric} trees. We establish a canonical correspondence between prefix-free codes and symmetric trees, preserving not only the lengths of codewords but also some additional commutative structure. Using this correspondence, we provide a result related to the commutative equivalence conjecture. We show that for every code, there exists a prefix-free code such that, for each fixed word length, the sums of powers of two determined by the occurrences of a distinguished symbol are equal. 
\end{abstract}


\section{Introduction}
\subsection{Some Background on Codes}

For basic exposition in coding theory, we refer the reader to~\cite{CodingTheoryExpo}. We introduce only the necessary notation and definitions. Throughout, we work over the binary alphabet $\{a,b\}$.

\begin{notation}
We denote by $\{a,b\}^{<\mathbb{N}}$ the free monoid generated by $\{a,b\}$; that is, the set of all finite words over $\{a,b\}$:
\[
\{a,b\}^{<\mathbb{N}}
=
\left\{
\langle x_1,\dots,x_n \rangle : n \ge 1,\ x_i \in \{a,b\}
\right\}
\cup \{\langle \rangle\},
\]
where $\langle \rangle$ denotes the empty word.

We use angle brackets to denote words when it is necessary to emphasize their sequence structure. For a word $w = \langle x_1,\dots,x_n \rangle \in \{a,b\}^{<\mathbb{N}}$, we define its length by $\lvert w \rvert = n$.

Given two words $w_1 = \langle x_1,\dots,x_n \rangle$ and $w_2 = \langle y_1,\dots,y_m \rangle$, their concatenation is defined by
\[
w_1 w_2 = \langle x_1,\dots,x_n, y_1,\dots,y_m \rangle.
\]
\end{notation}

\medskip

We now recall the notion of a (uniquely decodable finite) code. From this point onward, we shall refer to such object simply as code.

\begin{definition}
A \emph{code} is a finite set $C \subseteq \{a,b\}^{<\mathbb{N}}$ of finite words such that the following property holds: whenever
\[
w_1 \cdots w_n = w'_1 \cdots w'_m,
\]
with $w_1,\dots,w_n, w'_1,\dots,w'_m \in C$, then $n = m$ and $w_i = w'_i$ for all $1 \le i \le n$.
\end{definition}

\medskip

A fundamental constraint on the lengths of codewords is given by the classical Kraft–McMillan theorem. This result asserts not only that the lengths of codewords in any code must satisfy a certain constraint, but also that every sequence satisfying this condition arises from the lengths of codewords of a prefix-free code.
\begin{theorem}[Kraft--McMillan]
Let $C$ be a code over an alphabet of size $k$. Then
\[
\sum_{w \in C} k^{-\lvert w \rvert} \leq 1.
\]
Conversely, let $A$ be a finite multiset of positive integers satisfying
\[
\sum_{a \in A} k^{-a} \leq 1.
\]
Then there exists a prefix-free code $C$ such that the multiset of codeword lengths of $C$ is precisely $A$, that is, $\{ \lvert c \rvert : c \in C \} = A$ as multisets.
\end{theorem}

\medskip

In particular, every code is \emph{length-equivalent} to a prefix-free code; that is, for every code $C$ there exists a prefix-free code $C'$ together with a bijection $f \colon C \to C'$ such that $\lvert w \rvert = \lvert f(w) \rvert$ for all $w \in C$.

\medskip

Perrin and Sch\"{u}tzenberger introduced a stronger notion of equivalence~\cite{CECoriginalConjecture}. For a word $w = \langle x_1 ,\dots, x_n \rangle \in \{a,b\}^{n}$, we define
\[
[w]_a := \#\{1 \le i \le n : x_i = a\}, \qquad [w]_b := \#\{1 \le i \le n : x_i = b\}.
\]

\begin{definition}
Two codes $C_1, C_2$ are said to be \emph{commutatively equivalent} if there exists a bijection $f \colon C_1 \to C_2$ such that for every $w \in C_1$,
\[
[w]_a = [f(w)]_a \quad \text{and} \quad [w]_b = [f(w)]_b.
\]
\end{definition}

\medskip

Originally it was conjectured that every code is commutatively equivalent to a prefix-free code~\cite{CECoriginalConjecture}. However, this conjecture was disproved by Peter Shor~\cite{ShorCounterexample}, who found the code
\[
C_{\mathrm{Shor}} =
b\{1, a, a^7, a^{13}, a^{14} \}
\;\cup\;
\{a^3, a^8\}b\{1, a^2, a^4, a^6\}
\;\cup\;
a^{11}b\{1, a, a^2\},
\]
which is not commutatively equivalent to any prefix-free code. See~\cite{CompCounterexamples} for additional counterexamples. In light of this counterexample, the conjecture was restricted to a special class of codes~\cite{CECmaximalConjecture}, 

\begin{conjecture}[Sch\"{u}tzenberger]
Every \emph{maximal} code—that is, a code $C$ such that for every $w \in \{a,b\}^{<\mathbb{N}} \setminus C$, the set $C \cup \{w\}$ is not a code—is commutatively equivalent to a prefix-free code.
\end{conjecture}

The above conjecture is also referred to as the \emph{Commutative Equivalence Conjecture}. Partial results on the commutative
equivalence conjecture have been found, for example in~\cite{CECoriginalConjecture, PartialResult}. Our main result, Theorem~\ref{theorem:sum_of_powers_equivalence}, is somewhat related to this conjecture, as it establishes the existence of a prefix-free code that preserves certain quantities associated with symbol occurrences within the code. 

\subsection{Symmetric Trees and Main Results}

We now introduce the framework in which we consider trees. This formulation is consistent with the standard graph-theoretic notion of a rooted tree. the reader is referred to~\cite{GraphTheoryBasics} for basic terminology and definitions in graph theory. 

\begin{definition}
A \emph{rooted tree} is a finite, connected, acyclic graph $T$ together with a distinguished vertex $o \in V(T)$, called the \emph{root}. We denote by $\mathrm{dist}$ the usual graph distance on $T$.

Given vertices $u,v \in V(T)$, we say that $u$ is a \emph{child} of $v$ if $u$ and $v$ are adjacent and
\[
\mathrm{dist}(u,o) = \mathrm{dist}(v,o) + 1.
\]

For $v \in V(T) \setminus \{o\}$, we denote by $T_v$ the \emph{subtree rooted at $v$}, consisting of $v$ together with all its descendants.

A \emph{leaf} is a vertex without children---equivalently v is a leaf if and only if $T_v=\{v\}$ is a tree with a single vertex. For $k \in \mathbb{N}$, we define
\[
\mathcal{L}_k(T) := \{ v \in V(T) : v \text{ is a leaf and } \mathrm{dist}(v,o) = k \}.
\]
\end{definition}

\medskip

We next introduce a symmetry condition on rooted trees.

\begin{definition}
Let $T$ be a finite rooted tree with root $o$, and suppose that each vertex has at most three children (that is, $T$ is a finite subtree of the full $3$-ary tree).

We say that $T$ is \emph{symmetric} if for every vertex $v \in V(T)$ with at least two children, there exist distinct children $u_1,u_2$ of $v$ such that the rooted subtrees $T_{u_1}$ and $T_{u_2}$ are identical.
\end{definition}






We now state the main results, which justify the introduction of symmetric trees in the study of prefix-free codes.

\begin{theorem}\label{theorem:correspondence_symmetric_pfc}
There exists a bijective correspondence between prefix-free codes and symmetric trees. Moreover, under this correspondence, if $C$ is a prefix-free code and $T$ is the associated symmetric tree, then for every $k \in \mathbb{N}$,
\[
\lvert \mathcal{L}_k(T) \rvert
=
\sum_{\substack{w \in C \\ \lvert w \rvert = k}} 2^{[w]_a}.
\]
\end{theorem}

\medskip

We will use this correspondence to establish the following result:

\begin{theorem}\label{theorem:sum_of_powers_equivalence}
For every code $C$, there exists a prefix-free code $C'$ such that for every $k \in \mathbb{N}$, one has
\[
\sum_{\substack{w \in C \\ \lvert w \rvert = k}} 2^{[w]_a}
=
\sum_{\substack{w' \in C' \\ \lvert w' \rvert = k}} 2^{[w']_a}.
\]
\end{theorem}

\begin{remark}\label{remark:not same cardinality}
In general, the code $C'$ in Theorem~\ref{theorem:sum_of_powers_equivalence} need not have the same cardinality as $C$. Indeed, this fails in general, as demonstrated by Shor's counterexample. 

To see this, suppose that there existed a prefix-free code $C$ satisfying the conclusion of Theorem~\ref{theorem:sum_of_powers_equivalence} and such that $\lvert C \rvert = \lvert C_{\mathrm{Shor}} \rvert$. Then $C$ and $C_{\mathrm{Shor}}$ would necessarily be commutatively equivalent, contradicting Shor's result.
\end{remark}

\section{Proofs of Main Results}
\subsection{Proof of Theorem~\ref{theorem:correspondence_symmetric_pfc}}

We begin with the proof of the first direction, which follows ideas similar to those appearing in~\cite{PrefixCodesandFreeGroups}. To this end, we introduce several auxiliary notions.

\medskip

We first give an equivalent formulation of rooted labeled trees in terms of finite words. This formulation is consistent with the one given in~\cite{SolanBook}.

\begin{definition}
Let $\mathcal{A}$ be a finite alphabet. A \emph{labeled tree} over $\mathcal{A}$ is a nonempty set $T \subseteq \mathcal{A}^{<\mathbb{N}}$ of finite words such that $T$ is prefix-closed; that is, whenever $w \in T$ and $w'$ is a prefix of $w$, then $w' \in T$.

Elements of $T$ are called \emph{vertices}\footnote{Throughout, vertices—like codewords—are denoted using angle brackets $\langle \cdot \rangle$.}. Given $w, w' \in T$, we say that $w'$ is a \emph{child} of $w$ if $w'$ extends $w$ by exactly one symbol, that is,
\[
\lvert w' \rvert = \lvert w \rvert + 1 \quad \text{and} \quad w \text{ is a prefix of } w'.
\]

For $w \in T$, we define the \emph{subtree rooted at $w$} by
\[
T_w := \{\, w' \in T : w \text{ is a prefix of } w' \,\}.
\]
\end{definition}

\medskip

It is straightforward to verify that this definition is equivalent to the usual graph-theoretic notion of a rooted labeled tree.

\medskip

We next introduce the free group and its associated Cayley graph.

\begin{definition}
Let $\mathcal{A}$ be a finite set. The \emph{free group} $F(\mathcal{A})$ is the group consisting of all reduced words over the alphabet $\mathcal{A} \cup \mathcal{A}^{-1}$, where reduction means that no symbol appears adjacent to its inverse.

In the case $\mathcal{A} = \{a,b\}$, we write $F(\mathcal{A}) = F(a,b)$.
\end{definition}

\begin{definition}
The \emph{Cayley graph} of the free group $F(\mathcal{A})$, denoted $\mathrm{Cay}(F(\mathcal{A}))$, is the graph defined as follows:
\begin{enumerate}
    \item The vertex set is $F(\mathcal{A})$.
    \item For each $w \in F(\mathcal{A})$ and each $s \in \mathcal{A} \cup \mathcal{A}^{-1}$, there is an edge labeled $s$ connecting $w$ to $ws$.
\end{enumerate}
\end{definition}

We now complete the proof.

\medskip

Let $C$ be a prefix-free code over the alphabet $\{a,b\}$. Define
\[
Z_b(C)
:=
\left\{
\langle b, c_1, b, c_2, \dots, b, c_n \rangle
:\;
\langle c_1,\dots,c_n \rangle \in C
\right\},
\]
and let $T_b$ denote the associated labeled tree whose leaves correspond precisely to the elements of $C$.

\medskip

Consider the map
\[
\pi \colon \mathrm{Cay}(F(a,b)) \to \{a,b\}^{<\mathbb{N}},
\qquad
\pi\big(\langle c_1^{\varepsilon_1}, \dots, c_n^{\varepsilon_n} \rangle\big)
=
\langle c_1, \dots, c_n \rangle,
\]
where $c_i \in \{a,b\}$ and $\varepsilon_i \in \{\pm 1\}$ for all $i$.

\medskip

We consider the lifted subtree $\pi^{-1}(T_b)$ inside the Cayley graph $\mathrm{Cay}(F(a,b))$. It is shown in~\cite{PrefixCodesandFreeGroups} that the induced subtree consisting of vertices whose words have the symbol $b$ in every odd position corresponds to a prefix-free code over an alphabet of size $3$.

Let $T$ be the tree obtained by restricting this subtree to its even coordinates. More explicitly, $T$ consists of all words $\langle c'_1, \dots, c'_n \rangle$ over the (formally) free alphabet $\{a,a^{-1},b\}$ such that
\[
\langle b, c'_1, b, c'_2, \dots, b, c'_n \rangle
\in \pi^{-1}(T_b).
\]

\medskip

We claim that the underlying (unlabeled) tree $T$ is symmetric.

Indeed, let $w \in T$ be a vertex with at least two children. Then there exist distinct elements $s_1, s_2 \in \{a,b\} \cup \{a^{-1}, b^{-1}\}$ such that $ws_1, ws_2 \in T$. We first note that none of these extensions can involve $b^{-1}$, since this would produce a reduced word in $\pi^{-1}(T_b)$ containing a subword of the form $bb^{-1}$, which is impossible.

Moreover, at most one of $s_1, s_2$ can be equal to $b$. Consequently, without loss of generality, we may assume that $s_1 = a$. By symmetry of the construction, it follows that $wa^{-1} \in T$ as well. 

It follows that the subtrees rooted at $wa$ and $wa^{-1}$ are isomorphic (ignoring edge labels), and hence $T$ satisfies the symmetry condition.

\medskip

Finally, it is shown in~\cite{PrefixCodesandFreeGroups} that the number of leaves of $T$ that are mapped by $\pi$ to a given leaf
\[
\langle b, c_1, \dots, b, c_k \rangle \in T_b,
\qquad
\text{where } \langle c_1,\dots,c_k \rangle = w \in C,
\]
is precisely
\[
2^{\lvert \{ i \in [k] : c_i \neq b \} \rvert}
=
2^{[w]_a}.
\]
Therefore,
\[
\lvert \mathcal{L}_k(T) \rvert
=
\sum_{\substack{w \in C \\ \lvert w \rvert = k}} 2^{[w]_a},
\]
as required.

\medskip

We now prove the converse direction.

\medskip

Let $T$ be a symmetric rooted tree. We construct a labeling of $T$ by elements of the free monoid $\{a,a^{-1},b\}^{<\mathbb{N}}$.

\medskip

We proceed inductively. Assign to the root $o$ the empty word, denoted $\langle \rangle$.

\medskip

Suppose that a vertex $v \in V(T)$ has already been assigned a label $w_v \in \{a,a^{-1},b\}^{<\mathbb{N}}$. We describe how to label its children.

\begin{itemize}
    \item If $v$ has three children $u_1,u_2,u_3$, and two of them (say $u_1,u_2$) satisfy $T_{u_1} \cong T_{u_2}$, then assign
    \[
    w_{u_1} := w_v a, 
    \qquad
    w_{u_2} := w_v a^{-1}, 
    \qquad
    w_{u_3} := w_v b.
    \]

    \item If $v$ has exactly two children $u_1,u_2$ with $T_{u_1} \cong T_{u_2}$, then assign
    \[
    w_{u_1} := w_v a,
    \qquad
    w_{u_2} := w_v a^{-1}.
    \]

    \item If $v$ has a single child $u$, assign
    \[
    w_u := w_v b.
    \]
\end{itemize}

\medskip

This labeling defines an embedding of $T$ into the free monoid $\{a,a^{-1},b\}^{<\mathbb{N}}$, which we continue to denote by $T$ by a slight abuse of notation.

\medskip

Consider now the projection map
\[
\mathcal{P} \colon \{a,a^{-1},b\}^{<\mathbb{N}} \to \{a,b\}^{<\mathbb{N}},
\qquad
\mathcal{P}\big(\langle a \rangle\big) = \mathcal{P}\big(\langle a^{-1} \rangle\big) =a, \ \mathcal{P}\big(\langle b \rangle\big) =b, \ \mathcal{P}(w_1 \cdot w_2) = \mathcal{P}(w_1)\mathcal{P}(w_2)
\]
It is straightforward to verify that $\mathcal{P}(T)$ is a labeled tree whose set of leaves forms a prefix-free code $C$, such that,
\[
\lvert \mathcal{L}_k(T) \rvert
=
\sum_{\substack{w \in C \\ \lvert w \rvert = k}} 2^{[w]_a},
\]
as required. 
$\blacksquare$

\begin{remark}
The above construction shows, in particular, that a symmetric tree admits a canonical labeling by elements of $F(a,b)$, and hence determines a unique associated prefix-free code.
\end{remark}

\subsection{Proof of Theorem~\ref{theorem:sum_of_powers_equivalence}}
We begin by stating two auxiliary lemmas.

\begin{lemma}\label{lemma:mcmillan_3ary}
Let $C$ be a code. Then
\[
\sum_{c \in C} 3^{-\lvert c \rvert}\, 2^{[c]_a} \le 1.
\]
\end{lemma}

\begin{proof}
Using the notation from the proof of Theorem~\ref{theorem:correspondence_symmetric_pfc}, define
\[
\hat{C}
=
\left\{
\langle c'_1, \dots, c'_n \rangle :
\langle b, c'_1, b, c'_2, \dots, b, c'_n \rangle \in \pi^{-1}(T_b)
\right\}
\subseteq \{a,a^{-1},b\}^{<\mathbb{N}}.
\]
The key observation is that $\hat{C}$ forms a code over an alphabet of size~$3$.

\medskip

Indeed, suppose that a word admits two distinct factorizations into elements of $\hat{C}$:
\[
w_1 \cdots w_n = w'_1 \cdots w'_m,
\qquad
w_1,\dots,w_n,\, w'_1,\dots,w'_m \in \hat{C}.
\]
Applying the projection $\mathcal{P}$ to both sides yields two factorizations of the same word into elements of $C$. Since $C$ is a code, it follows that
\[
n = m
\quad \text{and} \quad
\mathcal{P}(w_i) = \mathcal{P}(w'_i)
\quad \text{for all } 1 \le i \le n.
\]
Since $w_1 \cdots w_n = w'_1 \cdots w'_n$,
we conclude that $w_i = w'_i$ for every $1 \le i \le n$, contradicting the assumption that the two factorizations are distinct. Thus, $\hat{C}$ is a code.

\medskip

The desired inequality now follows from McMillan's theorem applied to $\hat{C}$, viewed as a code over a $3$-letter alphabet, together with the same counting argument as in the proof of Theorem~3 in~\cite{PrefixCodesandFreeGroups}.
\end{proof}

\begin{lemma}\label{lemma:subset-sum-exact}
Let $A$ be a finite multiset of natural numbers, and let $N \in \mathbb{N}$. Suppose that 
\[
\sum_{n \in A} 2^n \ge 2^N \quad \text{and} \quad n \le N \text{ for all } n \in A.
\] 
Then there exists a sub-multiset $A' \subseteq A$ such that $\sum_{n \in A'} 2^n = 2^N$.
\end{lemma}

\begin{proof}
We proceed by induction on $N$. 

\textit{Base case:} $N = 0$.  
In this case, the only possible elements in $A$ are $0$, and the condition $\sum_{n \in A} 2^n \ge 2^0 = 1$ implies that $A$ contains at least one copy of $0$. Let $A'$ be a multiset consisting of exactly one $0$ from $A$. Then $\sum_{n \in A'} 2^n = 2^0 = 1$, as required.

\textit{Inductive step:} Consider $N > 0$, and assume that the statement holds for all natural numbers less than $N$. 

If $N \in A$, then we may take $A' = \{N\}$ and the result follows immediately. Otherwise, $N \notin A$. By assumption, $\sum_{n \in A} 2^n \ge 2^N \ge 2^{N-1}$. By the induction hypothesis applied to $N-1$, there exists a sub-multiset $A_1 \subseteq A$ such that $\sum_{n \in A_1} 2^n = 2^{N-1}.$ Consider the remaining multiset $A \setminus A_1$. We have 
\[
\sum_{n \in A \setminus A_1} 2^n = \sum_{n \in A} 2^n - \sum_{n \in A_1} 2^n \ge 2^N - 2^{N-1} = 2^{N-1}.
\] 
Applying the induction hypothesis again to $A \setminus A_1$ for $N-1$, we obtain a sub-multiset $A_2 \subseteq A \setminus A_1$ such that $\sum_{n \in A_2} 2^n = 2^{N-1}.$ Setting $A' = A_1 \cup A_2$, we have $\sum_{n \in A'} 2^n = 2^N$, as desired.
\end{proof}

\medskip

We now turn to the proof of the Main Result. Let $C$ be a code. By Theorem~\ref{theorem:correspondence_symmetric_pfc}, it suffices to construct a symmetric tree $T$ such that
\[
|\mathcal{L}_k(T)| = \sum_{\substack{w \in C \\ \lvert w\rvert=k}} 2^{[w]_a}.
\]

\medskip

We proceed by induction on $|C|$. 

\textbf{Base case:} If $|C| = 1$, the statement is trivial.

\medskip

\textbf{Inductive step:} Assume the statement holds for all codes of size less than $|C|$, and let $c_M \in C$ be a word of maximal length $M := \lvert c_M\rvert = \max\{|c| : c \in C\}$.

By the induction hypothesis, there exists a symmetric tree $T_0$ corresponding to the code $C \setminus \{c_M\}$, i.e.,
\[
|\mathcal{L}_k(T_0)| = \sum_{\substack{c \in C \setminus \{c_M\} \\ \lvert c\rvert = k}} 2^{[c]_a}.
\]

Let $Ex_M(T_0)$ denote the symmetric tree obtained by \emph{extending} each leaf of $T_0$ to depth $M$. Formally, for a leaf $v \in T_0$ at distance $k$ from the root, we attach a copy of the 3-ary tree of depth $M-k$ rooted at $v$. Equivalently, in terms of the canonical labeled tree, each leaf $c \in T_0$ is replaced by the set of new leaves
\[
\{ c \cdot w : w \in \{a, a^{-1}, b\}^{M - \lvert c\rvert} \}.
\]

By construction, all leaves of $Ex_M(T_0)$ are at distance $M$ from the root. Moreover, by Lemma~\ref{lemma:mcmillan_3ary}, we have
\[
|\mathcal{L}_M(Ex_M(T_0))| 
= \sum_{c \in C \setminus \{c_M\}} 3^{M-\lvert c\rvert} \cdot 2^{[c]_a} 
= 3^M \sum_{c \in C \setminus \{c_M\}} 3^{-\lvert c\rvert} \cdot 2^{[c]_a} 
\le 3^M - 2^{[c_M]_a}.
\]

Denote the complete 3-ary tree of depth $M$ by $T_M(3)$, and consider now the complement of the extended tree within $T_M(3)$:
\[
Ex_M(T_0)^C := T_M(3) \setminus Ex_M(T_0).
\]
Then
\[
|\mathcal{L}_M(Ex_M(T_0)^C)| = 3^M - |\mathcal{L}_M(Ex_M(T_0))| \ge 2^{[c_M]_a}.
\]

Furthermore, $Ex_M(T_0)^C$ is itself a symmetric tree. Therefore, it corresponds to a subset $S \subseteq \{a, b\}^M$ satisfying
\[
\sum_{w \in S} 2^{[w]_a} \ge 2^{[c_M]_a}.
\]

\medskip 
Observe that if $[w]_a \le [c_M]_a$ for all $w \in S$, then by Lemma~\ref{lemma:subset-sum-exact}, there exists a subset $S' \subseteq S$ such that
\[
\sum_{w \in S'} 2^{[w]_a} = 2^{[c_M]_a}.
\]
Adding the symmetric tree corresponding to $S'$ to $T_0$ then yields the desired symmetric tree for the full code $C$ (i.e. with the desired number of leaves), completing the inductive step. Here, by \say{adding}, we mean forming the symmetric tree corresponding to the union of $S'$ with the prefix free code associated to $T_0$.

\begin{remark}\label{remark: only b case}
    Also observe that, in general, at most one vertex at each depth can satisfy $[w]_a = 0$; furthermore, a code contains at most one word with this property. Thus, if we assume that $[c_M]_a = 0$, it follows, by parity considerations, that there necessarily exists $w \in S$ such that $[w]_a = 0$. In this case, we may take $S' = \{w\}$ as the desired subset.
\end{remark}

If the previous case does not hold, then there exists some $w \in S$ such that $[w]_a > [c_M]_a$. Let $i > [c_M]_a$ be minimal such that there exists $w \in S$ with $[w]_a = i$.

\medskip

Abusing notation and dentifying $T_0$ with the corresponding prefix-free code, we introduce the notation
\[
C_k(j) := \{ w \in T_0 : \lvert w\rvert = k , \; [w]_a = j \}, 
\quad 
w_k(j) := |C_k(j)|, \qquad k \in \mathbb{N} , 0\le j\le k
\]
i.e., $C_k(j)$ is the set of words in $T_0$ of length $k$ with exactly $j$ occurrences of the letter $a$.  Also denote $I := \{ |c| : c \in T_0\}$ to be the set of lengths of words in the prefix free code corresponding to $T_0$.
Observe that
\[
w_k(i) \le \binom{k}{i},
\]
since there are at most $\binom{k}{i}$ sequences of length $k$ with exactly $i$ letters equal to $a$.

\medskip

 Now, suppose there exists some $m \in I$, with $m < M$, such that for some $j$ satisfying
\[
\max(0, [c_M]_a - (M-m)) \le j \le \min(m, [c_M]_a),
\] 
we have 
\[
w_m(j) < \binom{m}{j}.
\]
Then there exists a “missing” word
\[
w \in \{a, b\}^m \setminus C_m(j), \quad \text{with } [w]_a = j.
\] 
We can extend this word to a new word of length $M$ by appending
\[
w' := w \cdot \langle \underbrace{a, \dots, a}_{[c_M]_a - j \text{ times}}, \underbrace{b, \dots, b}_{M-m - ([c_M]_a - j) \text{ times}} \rangle,
\]
which, by construction, has no prefix in $T_0$. Therefore, adding the symmetric tree corresponding to $w'$ to the tree $T_0$  then yields the desired symmetric tree for the full code
$C$, completing the inductive step.

\medskip

\begin{figure}[]
    \centering
    \includegraphics[width=0.9\linewidth]{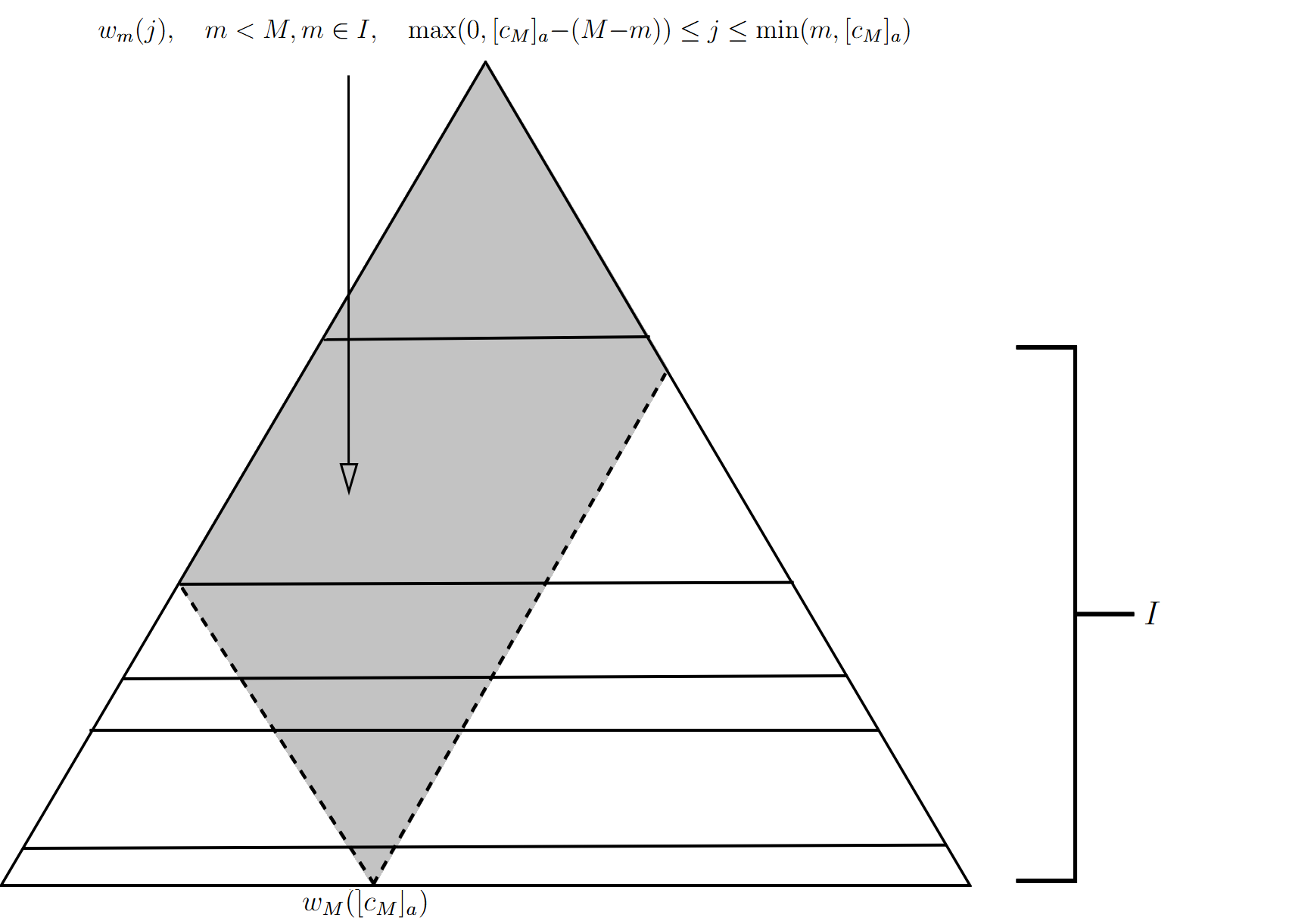}
    \caption{Region in which we seek a word that is neither a prefix of nor has a prefix in the existing code $T_0$.}
    \label{fig:placeholder}
\end{figure}

Otherwise, let $w \in S$ be such that $[w]_a = i$.  
By construction of $S$ there must exist some $m < M$ and a minimal index $j$ satisfying $\max(0, i - (M-m)) \le j \le \min(m, i)$ such that
\[
w_m(j) < \binom{m}{j}.
\]

By assumption, let $j' < j$ be maximal such that 
\[
\max(0, [c_M]_a - (M-m)) \le j' \le \min(m, [c_M]_a),
\] 
and
\[
w_m(j') = \binom{m}{j'}, \quad w_m(j'+1) = \binom{m}{j'+1}, \dots, w_m(j-1) = \binom{m}{j-1}.
\]

In this situation, we may perform the following adjustment: delete two words from $C_m(j-1)$ and add one of the missing elements to $C_m(j)$. Since $w_m(j) < \binom{m}{j}$, the operation is possible, and by construction, the total number of leaves of each depth in the symmetric tree remains unchanged. This can be seen by the fact that the total number of leafs of depth $m\in I$ is
\[
\sum_{k=0}^m 2^k w_m(k).
\] 

We continue this process iteratively. As long as $w_m(j') \ge 2$, the adjustment can be repeated without issue.  

If $w_m(j') = 1$, we must have $j' = 0$. In particular, by the maximality of $j'$, this would imply that $[c_M]_a=0$, which we already discussed in Remark~\ref{remark: only b case}.  

Consequently, after performing the described adjustments, the total number of leaves in the modified tree remains equal to that in $T_0$, and the “vacancy” in $C_m(j)$ has been \say{shifted} into a region where a word containing $[c_M]_a$ occurrences of $a$ can be added. Therefore, the configuration can be resolved as in the previous argument, completing the desired construction. 
$\blacksquare$

\section{Further Discussion}
We note that Theorem~\ref{theorem:sum_of_powers_equivalence} extends to codes over arbitrary finite alphabets. More precisely, we have the following statement.

\medskip
Let $C$ be a finite code over an alphabet $\mathcal{A}$ of size $A$. Then, for every $a \in \mathcal{A}$, there exists a prefix-free code $C'$ over $\mathcal{A}$ such that, for every $k \in \mathbb{N}$ with $k \ge A$\footnote{The restriction $k \ge A$ is imposed only to ensure that the condition $w_k(j) < A$ implies $j = 0$.},
\[
\sum_{\substack{w \in C \\ \lvert w \rvert = k}} A^{\,k - [w]_a}
=
\sum_{\substack{w' \in C' \\ \lvert w' \rvert = k}} A^{\,k - [w']_a}.
\]

We omit the proof in this general setting, as it follows by the same argument as in the binary case, which already captures the essential ideas. For example, in the general setting, Lemma~\ref{lemma:mcmillan_3ary} states that for codes over an alphabet $\mathcal{A}$ of size $A$, one has
\[
\sum_{w \in C} (2A - 1)^{-\lvert w \rvert} \, A^{\lvert w \rvert - [w]_a} \le 1.
\]
This is shown using the same argument presented above (together with the same counting argument as in the proof of Theorem~3 in~\cite{PrefixCodesandFreeGroups}).

\medskip

Also, in light of Remark~\ref{remark:not same cardinality}, it is natural to ask how close, in cardinality, a code can be to a corresponding prefix-free code satisfying the conclusion of Theorem~\ref{theorem:sum_of_powers_equivalence}. More precisely, given a code $C$, can one always find a prefix-free code $C'$ satisfying the conclusion of Theorem~\ref{theorem:sum_of_powers_equivalence} such that $\bigl|\, \lvert C \rvert - \lvert C' \rvert \,\bigr| = O(\lvert C \rvert)? $ or even $\bigl|\, \lvert C \rvert - \lvert C' \rvert \,\bigr| =O(1)?$

\end{document}